\documentclass[a4paper,11pt]{fullverllncs}
\usepackage[left=2.5cm,top=2.5cm,right=2.5cm,centering]{geometry}
\usepackage[dvipdfmx]{graphicx}
\usepackage{amsmath,amssymb,amsfonts}
\usepackage{ascmac}
\usepackage{array}
\usepackage{algpseudocode}
\usepackage{algorithm, algpseudocode}
\usepackage{multirow}
\usepackage{arydshln}
\usepackage{hhline} 

\usepackage[misc,geometry]{ifsym}
\usepackage{color}
\usepackage{rotating}

\usepackage[hyphens]{url}
\usepackage{tcolorbox}

\usepackage{tikz} 
\usetikzlibrary{trees}

\usepackage{xcolor}
\usepackage[colorlinks]{hyperref}
\definecolor{winered}{rgb}{0.5,0,0}
\definecolor{darkblue}{rgb}{0,0,0.5}
\definecolor{darkgreen}{rgb}{0,0.3,0}
\hypersetup{
linkcolor=winered,
citecolor=darkblue,
urlcolor=darkgreen
}

\newcommand{\A}{\mathsf{A}}

\newcommand{\C}{\mathsf{C}}
\newcommand{\R}{\mathsf{R}}

\newcommand{\negl}{\mathsf{negl}}

\newcommand{\Adv}{\mathsf{Adv}}

\newcommand{\Z}{\mathbb{Z}}

\newcommand{\bbL}{\mathbb{L}}
\newcommand{\N}{\mathbb{N}}

\newcommand{\w}{\mathsf{w}}

\newcommand{\sfGame}{\mathsf{G}}

\newcommand{\pp}{\mathsf{pp}}
\newcommand{\msk}{\mathsf{msk}}
\newcommand{\sk}{\mathsf{sk}}
\newcommand{\vk}{\mathsf{vk}}
\newcommand{\msg}{\mathsf{msg}}
\newcommand{\prf}{\mathsf{prf}}
\newcommand{\sff}{\mathsf{sff}}

\newcommand{\EUFAIDCMA}{\mathsf{EUF \mathchar`- AID \mathchar`-CMA}}
\newcommand{\EUFAPCMA}{\mathsf{EUF \mathchar`- AP \mathchar`- CMA}}
\newcommand{\EUFAPATPCMA}{\mathsf{EUF \mathchar`- AP\mathchar`- ATP \mathchar`- CMA}}
\newcommand{\Sign}{\mathsf{Sign}}
\newcommand{\Punc}{\mathsf{Punc}}

\newcommand{\inscorrupt}{\mathtt{corrupt}}

\newcommand{\Corrupt}{\mathsf{Corrupt}}
\newcommand{\NCL}{\mathsf{NCL}}
\newcommand{\GPV}{\mathsf{GPV}}

\newcommand{\id}{\mathsf{id}}
\newcommand{\HIBS}{\mathsf{HIBS}}
\newcommand{\nHIBSSetup}{\mathsf{Setup}}
\newcommand{\nHIBSExtract}{\mathsf{Extract}}
\newcommand{\nHIBSSign}{\mathsf{Sign}}
\newcommand{\nHIBSVerify}{\mathsf{Verify}}
\newcommand{\HIBSSetup}{\mathsf{HIBS.Setup}}
\newcommand{\HIBSExtract}{\mathsf{HIBS.Extract}}
\newcommand{\HIBSSign}{\mathsf{HIBS.Sign}}
\newcommand{\HIBSVerify}{\mathsf{HIBS.Verify}}
\newcommand{\ppHIBS}{\mathsf{pp}^{\mathsf{HIBS}}}
\newcommand{\skHIBS}{\mathsf{sk}^{\mathsf{HIBS}}}
\newcommand{\IBS}{\mathsf{IBS}}

\newcommand{\PPS}{\mathsf{PPS}}
\newcommand{\nPPSKGen}{\mathsf{KGen}}
\newcommand{\nPPSSign}{\mathsf{Sign}}
\newcommand{\nPPSPunc}{\mathsf{Punc}}
\newcommand{\nPPSVerify}{\mathsf{Verify}}
\newcommand{\PPSKGen}{\mathsf{PPS.KGen}}
\newcommand{\PPSSign}{\mathsf{PPS.Sign}}
\newcommand{\PPSPunc}{\mathsf{PPS.Punc}}
\newcommand{\PPSVerify}{\mathsf{PPS.Verify}}
\newcommand{\skPPS}{\mathsf{sk}^{\mathsf{PSS}}}
\newcommand{\vkPPS}{\mathsf{vk}^{\mathsf{PSS}}}

\newcommand{\CS}{\mathsf{CS}}
\newcommand{\Ours}{\mathsf{Ours}}

\spnewtheorem{assumption}{Assumption}{\bfseries}{\itshape}

\begin{document}

% \title{Prefix Puncturable Signatures from \\ Hierarchical Identity-Based Signatures}
% \author{}
% \authorrunning{}
% \titlerunning{Prefix Puncturable Signatures from HIBS}

%\institute{} 
%\email{}

\title{Prefix Puncturable Signatures with \\ Smaller Signing Key from HIBS\thanks{A preliminary version of this paper will appear at the 28th International Conference on Information and Communications Security (ICICS 2026).}}
\author{Masayuki Tezuka\inst{1,2}\textsuperscript{(\Letter)} \and Keisuke Tanaka\inst{2}}
\institute{Nagoya City University, Nagoya, Japan \and Institute of Science Tokyo, Tokyo, Japan\\
\email{m.tezuka@nsc.nagoya-cu.ac.jp}}
\authorrunning{M.Tezuka et al.}
\maketitle

\pagestyle{plain}
\noindent
\makebox[\linewidth]{September 20, 2026}

\begin{abstract}
Puncturable signatures, proposed by Bellare et~al. (EUROCRYPT 2016), allow a signing key to be punctured (updated) so that it loses the ability to sign particular messages while retaining the ability to sign all others.
Halevi et~al. (ASIACRYPT 2017) introduced prefix puncturable signatures, in which the signing key can be punctured with respect to a target prefix so that it cannot sign messages whose prefixes match the target prefix.

So far, several generic constructions of prefix puncturable signature schemes have been proposed, including constructions based on identity-based signatures (IBS) (ESORICS 2022) and delegated constrained signatures (IEEE Trans. Inf. Forensics Secure. 2024). 
However, these constructions suffer from drawbacks in terms of key size. 
When the prefix space is the set of all $\ell$-bit strings, the former construction requires a signing key consisting of $2^{\ell}$ IBS signing keys.
The latter construction, when instantiated with a lattice-based delegated constrained signature scheme, 
yields a punctured signing key whose size grows quadratically with the number of puncturing operations $Q^{\Punc}$.

In this paper, we present a generic construction of prefix puncturable signatures from hierarchical identity-based signature (HIBS) schemes. 
When the prefix space is $\{0,1\}^{\ell}$ and our construction is instantiated with the lattice-based HIBS scheme $\HIBS^{\GPV}$ by R{\"{u}}ckert (PQC 2010), our construction achieves a punctured signing key size bounded by $O(\ell Q^{\Punc})$.

\end{abstract}

\section{Introduction}

\paragraph{\bf (Prefix) Puncturable Signatures.}
Puncturable signatures, proposed by Bellare, Stepanovs, and Waters~\cite{BSW16}, are a special type of signatures that allow us to puncture (update) the signing key $\sk$.
For a target message $\msg^{*}$, the signing key $\sk$ can be punctured (updated) to obtain a punctured signing key $\sk_{\msg^{*}}$.
The punctured signing key $\sk_{\msg^{*}}$ loses the ability to sign the target message $\msg^{*}$ while retaining the ability to sign all other messages.
This functionality naturally allows us to view puncturable signatures as a special case of policy-based signatures \cite{BF14}, functional signatures \cite{BGI14}, or delegatable signatures \cite{BMS16}.

Halevi, Ishai, Jain, Komargodski, and Sahai \cite{HIJKSY17} proposed puncturable signatures that allow us to puncture $\sk$ with respect to a prefix $\prf^{*}$ of a message, rather than $\msg^{*}$ itself.
A signing key $\sk_{\prf^{*}}$ punctured with respect to $\prf^{*}$ allows us to sign any message $\msg$ whose prefix is different from $\prf^{*}$.
Guan and Zhandry~\cite{GZ21} referred to this type of scheme as a prefix puncturable signature scheme.

\paragraph{\bf Application of Puncturable Signatures.}
Prefix puncturable signatures are used as a building block for non-interactive multiparty computation~\cite{HIJKSY17}, disappearing signatures~\cite{GZ21}, blockchain protocols~\cite{LXFWZ20,JLSD24}, and privacy-aware data reporting for vehicular digital twin networks (VDTNs)~\cite{WMLZL25}.
Moreover, puncturable signatures provide forward secrecy at a fine-grained level \cite{JDS22}.

Forward-secure signatures provide a protection mechanism against key exposure by periodically updating the secret key.
This ensures that even if the secret key at the current time period is compromised, an adversary cannot forge signatures from past time periods, thereby achieving forward secrecy.
However, forward-secure signatures only provide protection at the granularity of time periods and do not support fine-grained control over individual messages or messages with particular features.
Puncturable signatures can be used to realize forward-secure signatures with revocation of the signing capacity for individual messages or messages with a pattern at a specific time interval.

\paragraph{\bf Previous Works on Puncturable Signatures.}
So far, several (prefix) puncturable signature schemes have been proposed.
The first puncturable signature scheme was given by Bellare et~al.~\cite{BSW16}.
They constructed the scheme from a one-way function and indistinguishability obfuscation (iO) \cite{BGIRSVY01}.
After their seminal work, Halevi et~al.~\cite{HIJKSY17} gave a prefix puncturable signature scheme by combining a non-interactive zero-knowledge (NIZK) proof system and a statistically binding commitment scheme.
Li, Xu, Fan, Wang, and Zhang~\cite{LXFWZ19,LXFWZ20} proposed a prefix puncturable signature scheme by combining a Bloom filter \cite{Blo70} and the Chinese IBS, an identity-based signature scheme standardized in ISO/IEC 14888-3 \cite{ISOIEC}.
Their scheme supports puncturing operations on a signing key, and its security is proven under the $\tau$-strong Diffie-Hellman ($\tau$-SDH) assumption \cite{BB04} in pairing groups.

Jiang, Duong, and Susilo \cite{JDS22} revisited the prefix puncturable signature scheme of Li et~al.~\cite{LXFWZ20} and pointed out several drawbacks arising from the use of Bloom filters. 
Concretely, hash collisions may cause false-positive errors, meaning that a signer may fail to sign even for non-punctured prefixes.
Moreover, the use of Bloom filters complicates the security proof because multiple keys may correspond to a single prefix. 
Motivated by these observations, Jiang et~al. proposed a generic construction from an identity-based signature (IBS) scheme without relying on Bloom filters.

After their work, Jiang, Li, Susilo, and Duong~\cite{JLSD24} gave a generic construction from a delegated constrained signature scheme. 
Shaw and Dutta~\cite{SD23} proposed an isogeny-based puncturable signature scheme from SQISign \cite{FKLPW20}.

\subsection{Motivation}
\paragraph{\bf (Punctured) Signing Key Size.}
Several proposed schemes suffer from drawbacks.
The constructions \cite{BSW16} and \cite{HIJKSY17} rely on iO or NIZK, and the resulting schemes are impractical due to heavy computation.  
The construction \cite{LXFWZ20} is a pairing-based scheme that is vulnerable to quantum computers.  
The generic constructions \cite{JDS22} and \cite{JLSD24} allow us to obtain schemes that are resistant to quantum computers.  
However, these generic constructions have the drawback of large (punctured) secret key size.
Here, we consider the key sizes of \cite{JDS22} and \cite{JLSD24} when instantiated with lattice-based primitives.

\paragraph{\bf Key Size of Lattice-Based Instantiations.}
When the prefix space is the set of all $\ell$-bit strings, the generic construction \cite{JDS22} requires a signing key to contain $2^{\ell}$ IBS signing keys.  
A punctured signing key consists of $(2^{\ell} - Q^{\Punc})$ IBS signing keys, where $Q^{\Punc}$ denotes the number of puncturing operations applied to the signing key.  
The generic construction \cite{JLSD24} requires a signing key to contain a single signing key of a delegated constrained signature scheme.
A punctured signing key also consists of a single delegated constrained signing key.  

Jiang et~al. \cite{JLSD24} provided a comparison table between lattice-based instantiations of the puncturable signature schemes in \cite{JLSD24} and \cite{JDS22}.  
They compared the schemes obtained from the generic construction of \cite{JDS22} instantiated with an efficient IBS by Tian and Huang \cite{TH16}, and the generic construction of \cite{JLSD24} instantiated with a delegated constrained signature scheme by Tsabary \cite{Tsa17}.  
We highlight this table with respect to key size in Fig.~\ref{Fig_comparion_PPS}.

\begin{figure}[htbp]
\begin{center}
\scalebox{1}{
\begingroup
\renewcommand{\arraystretch}{1.4}
\tabcolsep = 1.5pt
\begin{tabular}{lccccc}
\hline
Scheme
&
$\sk$
&
$\sk_{\prf}$
&
$\vk$

\\
\hline
\hline
\begin{tabular}{l}
\cite{JDS22}
\end{tabular}
&
\begin{tabular}{l}
$2^{\ell} \cdot \Z_{q}^{m \times m}$
\end{tabular}
&
\begin{tabular}{l}
$(2^{\ell} - Q^{\Punc}) \cdot \Z_{q}^{m  \times n } $
\end{tabular}
&
\begin{tabular}{l}
$\Z_{q}^{n \times m}$
\end{tabular}

\\
\hline
\begin{tabular}{l}
\cite{JLSD24}
\end{tabular}
&
\begin{tabular}{l}
$\Z_{q}^{m \times m}$
\end{tabular}
&
\begin{tabular}{l}
$\Z_{q}^{(m + Q^{\Punc}  \cdot n \lceil \log{q} \rceil) \times (m + Q^{\Punc} \cdot n\lceil \log{q} \rceil)} $
\end{tabular}
&
\begin{tabular}{l}
$\Z_{q}^{n \times (m + (\ell + 1) \cdot n  \lceil \log{q} \rceil )}$
\end{tabular}

\\
\hline
\begin{tabular}{l}
$\PPS_{\Ours}$\\
$[\HIBS^{\GPV}]$
\end{tabular}
&
\begin{tabular}{l}
$\Z_{q}^{m \times m}$
\end{tabular}
&
\begin{tabular}{l}
$O(\ell Q^{\Punc}) \cdot \Z_{q}^{m \times m}$
\end{tabular}
&
\begin{tabular}{l}
$\Z_{q}^{n \times m_{1} + m_{2}} + 2 \ell \cdot \Z_{q}^{n \times  m_{2}}$
\end{tabular}

\\

\hline
\end{tabular}
\endgroup
}
\end{center}
\caption{\small Key size comparison among lattice-based puncturable signature schemes on the prefix space $\{0,1\}^{\ell}$.}
In the column ``Scheme'', \cite{JDS22} denotes the scheme instantiated the generic construction \cite{JDS22} with the IBS in~\cite{TH16}.
\cite{JLSD24} denotes the generic construction \cite{JLSD24} with the delegated constrained signature scheme by Tsabary \cite{Tsa17}.
$\PPS_{\Ours}[\HIBS^{\GPV}]$ denotes the scheme obtained from our generic construction of a puncturable signature scheme using the hierarchical identity-based signature scheme $\HIBS^{\GPV}$ by R{\"{u}}ckert \cite{Ruc10}.
The columns ``$\sk$'', ``$\sk_{\prf}$'', and $\vk$ denote a signing key, a punctured signing key, and a verification key, respectively.
The parameter $Q^{\Punc}$ denotes the number of puncturing operations applied to a signing key, and $\ell$ denotes the prefix length.
The same parameter settings for $q$, $n$, and $m = n \lceil \log{q} \rceil $ are used in these schemes.
These parameters are determined by the same lattice trapdoor generation algorithm.
The parameters $m_{1}, m_{2} > 0$ satisfy $m_{1} + m_{2} = m$.

\label{Fig_comparion_PPS}
\end{figure}

From Fig.~\ref{Fig_comparion_PPS}, we observe that the lattice-based puncturable signature scheme obtained from the generic construction~\cite{JLSD24} instantiated with Tsabary’s delegated constrained signature scheme~\cite{Tsa17} suffers from a drawback in the punctured signing key size.  
In particular, the punctured signing key size grows quadratically in the number of puncturing operations $Q^{\Punc}$.  
Such quadratic growth in $Q^{\Punc}$ is undesirable in some applications.

For example, we consider an application of a puncturable signature scheme to a forward-secure signature scheme \cite{JDS22}.  
In this application, message prefixes in the puncturable signature scheme are interpreted as time periods.  
Then, by puncturing the signing key with respect to the prefix corresponding to the current time period and updating the signing key accordingly, the punctured signing key realizes the key-update mechanism of a forward-secure signature scheme.  
If the maximum number of time periods is $T$, then $T$ puncturing operations are required.  
Consequently, it is undesirable for the punctured signing key size to grow as $\Omega(T^{2})$.

\paragraph{\bf Current Open Question.} 
To summarize the above facts, the following remains open in lattice-based prefix puncturable signatures:
\begin{quote}
{\it Is it possible to construct a lattice-based prefix-puncturable signature scheme with prefix space $\{0,1\}^{\ell}$ whose punctured signing key size grows subquadratically in $Q^{\Punc}$ and sublinearly in the size of the prefix space~$2^\ell$?}
\end{quote}

\subsection{Our Result}
\paragraph{\bf Main Result.}
In this work, we give an affirmative answer to this question.  
More precisely, we present a generic construction of a prefix puncturable signature scheme $\PPS_{\Ours}[\HIBS]$ from a hierarchical identity-based signature (HIBS) scheme $\HIBS$.  
If we instantiate our construction with a lattice-based adaptively secure HIBS scheme, we obtain a lattice-based puncturable signature scheme.

\paragraph{\bf Key Size Comparison.}
For our lattice-based instantiation, we consider the lattice-based HIBS scheme $\HIBS^{\GPV}$ \cite{Ruc10}.  
The scheme $\HIBS^{\GPV}$ is based on the signature scheme by Gentry, Peikert, and Vaikuntanathan \cite{GPV08}, and its security is proven under the short integer solution (SIS) problem in the random oracle model (ROM).  
We denote the instantiation of our scheme $\PPS_{\Ours}$ with $\HIBS^{\GPV}$ as $\PPS_{\Ours}[\HIBS^{\GPV}]$.  

The key size comparison between lattice-based instantiations of $\PPS_{\Ours}$ and previous works for prefix space $\{0,1\}^{\ell}$ is given in Fig.~\ref{Fig_comparion_PPS}.  
In our instantiation $\PPS_{\Ours}[\HIBS^{\GPV}]$, the punctured signing key size is $O(\ell Q^{\Punc})$, where $Q^{\Punc}$ is the number of puncturing operations applied to the original signing key.
When the number of puncturing operations increases, our scheme achieves a smaller punctured signing key size than the lattice-based instantiation of \cite{JLSD24}.

\paragraph{\bf Future Directions.}
Our construction is generic and can be instantiated with any adaptively secure HIBS scheme.
Constructing more efficient lattice-based HIBS schemes that are adaptively secure in the QROM and obtaining more efficient prefix-puncturable signature schemes are interesting open problems.

\subsection{How to Obtain Our Scheme}
\paragraph{\bf Starting Point.}
We explain the main idea of our construction.  
Our construction is obtained by building on the generic construction of Jiang et~al.~\cite{JDS22}.  
Their construction uses an IBS scheme.  
We briefly recall the idea of their construction.

To simplify the exposition, we consider constructing a prefix-puncturable signature scheme with prefix space $\{0,1\}^{\ell}$ from their generic construction.
In their construction, the key generation algorithm generates the public parameters and a master secret key $(\pp, \msk)$ of $\IBS$ by running the setup algorithm of $\IBS$.  
Next, for all $\prf \in \{0,1\}^{\ell}$, the algorithm generates the signing key $\sk_{\prf}$ corresponding to $\prf$ by running the key-extraction algorithm of $\IBS$.
Then, the algorithm deletes $\msk$ and returns a list $L$ of $\IBS$ secret keys as the signing key $\sk$ of the prefix puncturable signature scheme, where $\sk_{\prf}$ is stored in $L[\prf]$ for all $\prf \in \{0,1\}^{\ell}$.  
The puncturing algorithm takes a signing key $\sk = L$ and a prefix $\prf^{*}$.  
It then updates the list entry as $L[\prf^{*}] \leftarrow \bot$ and returns the updated list $L$ as the new signing key $\sk'$.  
In their approach, the signing key consists of $2^{\ell}$ secret keys of the underlying IBS scheme, which is undesirable in practice.

\paragraph{\bf Our Solution.}
To address this drawback, we use a hierarchical identity-based signature (HIBS) scheme instead of an IBS scheme. 
We consider a level-$\ell$ HIBS scheme with identity space $ID = ID_{1} \times ID_{2} \times \cdots \times ID_{\ell} = \{0, 1\}^{\ell}$ (i.e., the level $i$ identity space $ID_{i} = \{0,1\}$).
For this tree, the root is labeled with the empty string $\epsilon$.
For every node at level $i$ labeled by $x \in \{0, 1\}^{i}$, its left and right children are labeled by $x0$ and $x1$, respectively.

In our construction of a prefix puncturable signature scheme, a signing key $\sk$ is the master secret key $\msk$ of the HIBS scheme.
Let $\sk_{t}$ be the current (punctured) signing key and $D$ be a set of punctured prefixes so far. 
To puncture (revoke) a prefix $\prf$ from a signing key $\sk_{t}$, we compute a set $C$ whose elements correspond to the roots of subtrees such that:
\begin{itemize}
\item Every leaf whose label is not in $D \cup \{\prf\}$ lies in the subtree of at least one node in~$C$.
\item No leaf whose label is in $D \cup \{\prf\}$ lies in the subtree of any node in $C$.
\end{itemize}
That is, $C$ is a set of labels of roots of subtrees that cover all leaves except the leaves whose labels are contained in $D\cup\{\prf\}$.
By deriving the HIBS signing keys corresponding to labels in $C$ from $\sk_{t}$ by running the key extraction algorithm of the HIBS scheme, we obtain the punctured key $\sk_{t+1} = \{\sk_{\prf}\}_{\prf \in C}$.
This operation can be performed by running the cover set algorithm of Naor et~al. \cite{NNL01}, and the number of elements in $C$ is bounded by $O(\ell Q^{\Punc})$, where $Q^{\Punc}$ denotes the number of puncturing operations.
As a result, we obtain a generic construction whose (punctured) signing key consists of $O(\ell Q^{\Punc})$ HIBS signing keys.

\subsection{Related Works}
\paragraph{\bf (Hierarchical) Identity-Based Signatures.}
The concept of identity-based signatures (IBS) was introduced by Shamir \cite{Sha84}.
In an IBS system, a trusted authority called the Key Generation Center (KGC) generates a public parameter $\pp$ and a master secret key $\msk$.
When a signer with identity $\id$ joins the system, the KGC uses the master secret key $\msk$ to generate a signing key $\sk_{\id}$ associated with $\id$ and sends it to the signer.
Then, the signer signs a message $\msg$ with $\sk_{\id}$ and generates a signature $\sigma$.
The generated signature $\sigma$ is verified using the public parameters $\pp$ and the signer’s identity $\id$.
A strength of IBS is that it reduces the number of public keys that need to be managed and significantly simplifies key management, thereby removing the overhead associated with traditional public-key infrastructures.
On the other hand, IBS has a weakness that the KGC must generate signing keys for all users using the master secret key $\msk$, which makes it a scalability bottleneck in large-scale systems. 

Hierarchical identity-based signatures (HIBS), proposed by Gentry and Silverberg \cite{GS02}, extend IBS by organizing identities into a hierarchical structure, enabling signing keys to be delegated from higher-level entities to lower-level entities. 
In a HIBS system, an entity holding a signing key for an identity at level $k$ can derive signing keys for its descendant identities. 
This allows key generation to be delegated to lower-level entities without involving the root authority for every user.

\paragraph{\bf Security of Hierarchical Identity-Based Signatures.}
The unforgeability of (H)IBS schemes is considered under two security notions: selective-ID security and adaptive-ID security.  
In the selective-ID setting, the forger is required to declare a target identity $\id^{*}$ before a public parameter $\pp$ is given. 
The adaptive-ID setting allows the forger to choose the target identity $\id^{*}$ after seeing the public parameter $\pp$ and making signing and key corruption queries.
The adaptive-ID security is stronger than the selective-ID security.
In this work, we use an adaptive-ID secure HIBS for our construction.

In previous work, several HIBS schemes satisfying adaptive-ID security have been proposed.
Kiltz, Mityagin, Panjwani, and Raghavan \cite{KMPR05}  proposed a generic construction of an adaptive-ID secure HIBS scheme from an append-only signature scheme.
They also gave a generic construction of an append-only signature scheme from a digital signature scheme.
As a result, we can obtain an adaptive-ID secure HIBS scheme from a digital signature scheme.

R{\"u}ckert \cite{Ruc10} proposed two selective-ID secure lattice-based HIBS schemes.  
One scheme is based on the lattice-based signature scheme by Gentry et~al.~\cite{GPV08}, and its security is proven under the hardness of the SIS problem in the random oracle model (ROM).  
The other scheme is based on the lattice-based signature scheme by Cash, Hofheinz, Kiltz, and Peikert \cite{CHKP10}, and its security is proven under the hardness of the SIS problem without the ROM.  
Moreover, they also constructed adaptive-ID secure HIBS schemes by combining each lattice-based scheme with a chameleon hash function \cite{KR00}.  
For our lattice-based instantiation, we use the HIBS scheme $\HIBS^{\GPV}$, which is obtained by combining the former selective-ID secure scheme with a chameleon hash function.

\section{Preliminaries}\label{Sec_Prelimi}
In this section, we introduce notation and review the definition of a hierarchical identity-based signature scheme and its security notion.

\subsection{Notations}
We introduce the notation used throughout this paper.
Let $1^{\lambda}$ be the security parameter. 
A function $f(\lambda)$ is negligible in $\lambda$ if $f(\lambda)$ tends to $0$ faster than $\frac{1}{\lambda^c}$ for every constant $c > 0$.
We write $f(\lambda)=\negl(\lambda)$ to indicate that $f$ is negligible in $\lambda$.
For a positive integer $n$, we define a set $[n]: =\{1,\dots, n\}$.
For a finite set $S$, $s \xleftarrow{\$} S$ denotes that $s$ is chosen from $S$ uniformly at random.
For finite sets $S$ and $T$, we denote by $S \backslash T$ the set obtained by removing the elements of $T$ from $S$.
We denote the set of arbitrary-length bit strings by $\{0, 1\}^{*}$.
For strings $s$ and $t$, we denote the concatenation of these strings by $s||t$.
For a list $\bbL$, $|\bbL|$ represents the number of elements in $\bbL$.
For an algorithm $\A$, we write $y \leftarrow \A(x)$ to denote that $\A$ outputs $y$ on input~$x$.
We abbreviate probabilistic polynomial time as PPT.

\subsection{Hierarchical Identity-Based Signatures}\label{SubSect_HIBE}
We review a definition of an $\ell$-level hierarchical identity-based signature scheme and its security notion.

\begin{definition}[Hierarchical Identity-Based Signature Scheme]
An $\ell$-level hierarchical identity-based signature scheme $\HIBS$ with an identity space $ID = ID_{1} \times ID_{2} \times \dots \times ID_{\ell}$ and message space $M$ is a tuple of algorithms $(\nHIBSSetup, \nHIBSExtract, \allowbreak \nHIBSSign, \nHIBSVerify)$.
\begin{itemize}

\item $\nHIBSSetup (1^{\lambda}):$ A setup algorithm takes as an input a security parameter~$1^{\lambda}$. 
It returns a public parameter $\pp$ and a master signing key $\sk_{\epsilon}$.

\item $\nHIBSExtract (\pp, \sk_{\id_{1}||\dots||\id_{j-1}}, \id_{j}):$ A key extraction algorithm takes as an input a public parameter $\pp$, a signing key $\sk_{\id_{1}||\dots||\id_{j-1}}$ for an identity (prefix) $\id_{1}||\dots||\id_{j-1}$, and an identity $\id_{j}$ where $j \in [\ell]$. 
It returns a signing key $\sk_{\id}$ for an identity (prefix) $\id_{1}||\dots||\id_{j-1}||\id_{j}$.
(In the case of $j=1$, $\nHIBSExtract$ takes a tuple $(\pp, \sk_{\epsilon}, \id_{1})$ and outputs a signing key $\sk_{\id_{1}}$. For $j \leq \ell$, we call $\id_{1}||\dots||\id_{j-1}$ an identity prefix.)
\item $\nHIBSSign (\pp, \sk_{\id}, \id, \msg):$ A signing algorithm takes as an input a public parameter $\pp$, a signing key $\sk_{\id}$, an identity $\id$, and a message $\msg \in M$.
It returns a signature $\sigma$.

\item $\nHIBSVerify (\pp, \id = \id_{1}|| \dots || \id_{\ell}, \msg, \sigma):$ A verification algorithm takes as an input a public parameter $\pp$, an identity $\id= \id_{1}|| \dots || \id_{\ell}$, a message $\msg$, and a signature $\sigma$.
It returns a bit $b \in  \{0, 1\}$.

\end{itemize}
\end{definition}

We require $\HIBS$ to satisfy the following correctness.
\paragraph{\bf Correctness.}
An $\ell$-level hierarchical identity-based signature scheme $\HIBS =(\nHIBSSetup, \nHIBSExtract, \allowbreak \nHIBSSign, \nHIBSVerify)$ satisfies correctness if $\forall \lambda \in \N$, $\forall \id_{i} \in ID_{i}$ for $i \in [\ell]$, $(\pp,\sk_{\epsilon}) \leftarrow \nHIBSSetup (1^{\lambda})$, $\sk_{\id_{1}} \leftarrow \nHIBSExtract (\pp, \sk_{\epsilon}, \id_{1})$, $\sk_{\id_{1}||\dots||\id_{i-1}||\id_{i}} \leftarrow \nHIBSExtract (\pp, \allowbreak \sk_{\id_{1}||\dots||\id_{i-1}}, \id_{i})$ for $i \in \{2, \dots, \ell\}$, $\forall \msg \in M$, $\id = \id_{1}||\dots||\id_{\ell}$, and $\sigma \leftarrow \nHIBSSign (\pp, \sk_{\id}, \id, \msg)$, $\Pr[\nHIBSVerify (\pp, \id, \msg, \sigma)] = 1 - \negl(\lambda)$ holds, where the probability is taken over the randomness of $\nHIBSSetup$, $\nHIBSExtract$, $\nHIBSSign$, and $\nHIBSVerify$.

We review the security definition of an $\ell$-level hierarchical identity-based signature scheme.
To define the security, we introduce some notations.
For an identity prefix $\id_{1}||\dots||\id_{i}$, we define the set $I[\id_{1}||\dots||\id_{i}]$ as 
\begin{equation*}
I[\id_{1}||\dots||\id_{i}]:= 
\left\{\id_{1}||\dots||\id_{j} \middle| 
\begin{split}
&j \in \{i, \dots, \ell \}, \\
&\id_{i+1} \in ID_{i+1}, \dots, \id_{j} \in ID_{j}\\
\end{split}
\right\}.
\end{equation*}
That is, $I[\id_{1}||\dots||\id_{i}]$ denotes the set consisting of the identity $\id_{1}||\dots||\id_{i}$ itself and all its descendant identities.
\begin{definition}[EUF-AID-CMA Security \cite{KMPR05}]
Let $\HIBS =(\nHIBSSetup, \nHIBSExtract, \allowbreak \nHIBSSign, \allowbreak \nHIBSVerify)$ be an $\ell$-level hierarchical identity-based signature scheme and $\A$ be a PPT adversary.
The existential unforgeability under chosen message attacks with adaptive identity (EUF-AID-CMA) security is defined via the following EUF-AID-CMA game $\sfGame^{\EUFAIDCMA}_{\HIBS, \A}(1^{\lambda})$ between a challenger $\C$ and the adversary $\A$.

\begin{itemize}
\item {\bf Initial Setup:}
$\C$ initializes lists $\bbL^{\Sign} \leftarrow \{\}$, $\bbL^{\Corrupt} \leftarrow \{\}$, runs $(\pp,\sk_{\epsilon}) \leftarrow \nHIBSSetup (1^{\lambda})$, and sends $\pp$ to $\A$.

\item {\bf Query Phase:}
$\A$ makes the following corruption queries and signing queries polynomially many times in an arbitrary order.
\begin{itemize}
\item  {\bf Corruption query:}
For a corruption query on $\id_{1}||\dots||\id_{j}$ where $j \in [\ell]$, if there is an identity prefix $ \widetilde{\id}_{1} || \dots || \widetilde{\id}_{k} \in \bbL^{\Corrupt}$ such that $\id_{1}||\dots||\id_{j} \allowbreak \in I[\widetilde{\id}_{1} || \dots || \widetilde{\id}_{k}]$, $\C$ returns $\bot$.
Otherwise $\C$ updates $\bbL^{\Corrupt} \leftarrow \bbL^{\Corrupt} \cup \{\id_{1}||\dots||\id_{j} \}$, runs $\sk_{\id_{1}} \leftarrow \nHIBSExtract (\pp, \sk_{\epsilon}, \id_{1})$ and $\sk_{\id_{1}||\dots||\id_{i-1}||\id_{i}} \leftarrow \nHIBSExtract (\pp, \allowbreak \sk_{\id_{1}||\dots||\id_{i-1}}, \id_{i})$ for $i \in \{2, \dots, j\}$.
$\C$ returns $\sk_{\id_{1}||\dots||\id_{j}}$ to~$\A$.

\item  {\bf Signing query:}
For a signing query on $(\msg, \id = \id_{1}||\dots||\id_{\ell})$, if there is an identity prefix $ \widetilde{\id}_{1} || \dots || \widetilde{\id}_{k} \in \bbL^{\Corrupt}$ such that $\id \in I[\widetilde{\id}_{1} || \dots || \widetilde{\id}_{k}]$, $\C$ returns $\bot$.
Otherwise $\C$ updates $\bbL^{\Sign} \leftarrow \bbL^{\Sign} \cup \{(\id, \msg)\}$, runs $\sk_{\id_{1}} \leftarrow \nHIBSExtract (\pp, \sk_{\epsilon}, \id_{1})$, $\sk_{\id_{1}||\dots||\id_{i-1}||\id_{i}} \leftarrow \nHIBSExtract (\pp, \allowbreak \sk_{\id_{1}||\dots||\id_{i-1}}, \id_{i})$ for $i \in \{2, \dots, \ell \}$, $\sigma \leftarrow \nHIBSSign (\pp, \sk_{\id}, \id, \msg)$.
$\C$ returns $\sigma$ to~$\A$.
\end{itemize}

\item {\bf Finalization:}  
$\A$ finally outputs a forgery $(\id^{*}=\id^{*}_{1}||\dots||\id^{*}_{\ell} ,\msg^{*}, \sigma^{*})$ to $\C$. 
If there is an identity prefix $ \widetilde{\id}_{1} || \dots || \widetilde{\id}_{k} \in \bbL^{\Corrupt}$ such that $\id^{*} \allowbreak \in I[\widetilde{\id}_{1} || \dots || \widetilde{\id}_{k}]$, return $0$.
If $(\id^{*},\msg^{*}) \in \bbL^{\Sign}$, return $0$.
If $\nHIBSVerify (\pp, \id^{*}, \msg^{*}, \allowbreak \sigma^{*}) \allowbreak = 1$ return~$1$. 
Otherwise, return $0$.
\end{itemize}

The advantage of $\A$ is defined as $\Adv^{\EUFAIDCMA}_{\HIBS, \A}(\lambda):= \Pr[\sfGame^{\EUFAIDCMA}_{\HIBS, \A}(1^{\lambda}) \Rightarrow 1]$.
We say that $\HIBS$ satisfies the EUF-AID-CMA security if for any PPT adversary $\A$, $\Adv^{\EUFAIDCMA}_{\HIBS, \A}(\lambda)$ is $\negl(\lambda)$.
\end{definition}

\section{Prefix Puncturable Signatures}\label{Sect_PPS}
In this section, we review the definition of a prefix puncturable signature scheme and its security notion.
Then, we revisit the security notion and propose a new security model.

\subsection{Prefix Puncturable Signature Scheme}
We review a definition of a prefix puncturable signature scheme and its security notion.

\begin{definition}[Prefix Puncturable Signature Scheme]
A prefix puncturable signature scheme $\PPS$ with a message space $M = P \times S$ is a tuple of algorithms $(\nPPSKGen, \nPPSPunc, \allowbreak \nPPSSign, \nPPSVerify)$, where $P$ denotes the prefix space and $S$ denotes the suffix space.
\begin{itemize}

\item $\nPPSKGen (1^{\lambda}):$ A key generation algorithm takes as an input a security parameter~$1^{\lambda}$. 
It returns a verification key $\vk$ and a signing key $\sk_{0}$.

\item $\nPPSPunc (\vk, \sk_{i}, \prf):$ A puncturing algorithm takes as an input a current signing key $\sk_{i}$ and a prefix $\prf \in P$. 
It returns an updated punctured signing key $\sk_{i+1}$.
\item $\nPPSSign (\vk, \sk_{i}, \msg=\prf||\sff):$ A signing algorithm takes as an input a verification key, a signing key $\sk_{i}$ and a message $\msg=\prf||\sff \in M$. 
It returns a signature $\sigma$.

\item $\nPPSVerify (\vk, \msg=\prf||\sff, \sigma):$ A verification algorithm takes as an input a verification key $\vk$, a message $\msg=\prf||\sff \in M$, and a signature $\sigma$.
It returns a bit $b \in  \{0, 1\}$.

\end{itemize}
\end{definition}

We require $\PPS$ to satisfy the following correctness.
\paragraph{\bf Correctness.}
A prefix puncturable signature scheme $\PPS = (\nPPSKGen, \nPPSPunc, \allowbreak \nPPSSign, \nPPSVerify)$ satisfies correctness if $\forall \lambda \in \N$, $\forall t \in \N \cup \{0\}$, $(\vk, \sk_{0}) \leftarrow \nPPSKGen (1^{\lambda})$, $\forall \prf_{1}, \dots, \prf_{t} \in P$, $\sk_{i} \leftarrow \nPPSPunc (\vk, \sk_{i-1}, \prf_{i})$ for $i \in [t]$, $\forall \prf \in P \backslash \{ \prf_{1}, \dots, \prf_{t} \}$, $\forall \sff \in S$, $\sigma \leftarrow \nPPSSign (\vk, \sk_{t}, \msg=\prf||\sff)$, $\Pr[\nPPSVerify (\vk, \msg=\prf||\sff, \sigma)]= 1 - \negl(\lambda)$, where  the probability is taken over the randomness of $\nPPSKGen$, $ \nPPSPunc$, $\nPPSSign$, and $\nPPSVerify$.

\begin{definition}[EUF-AP-CMA Security \cite{JDS22}]\label{Def_PPS_EUF-AP-CMA}
Let $\PPS = (\nPPSKGen, \allowbreak \nPPSPunc, \allowbreak \nPPSSign, \allowbreak \nPPSVerify)$ be a prefix puncturable signature scheme and $\A$ be a PPT adversary.
The existential unforgeability under chosen message attacks with adaptive puncturing (EUF-AP-CMA) security is defined via the following EUF-AP-CMA game $\sfGame^{\EUFAPCMA}_{\PPS, \A}(1^{\lambda})$ between a challenger $\C$ and the adversary $\A$.

\begin{itemize}
\item {\bf Initial Setup:}
$\C$ initializes lists $\bbL^{\Sign} \leftarrow \{\}$, $\bbL^{\Punc} \leftarrow \{\}$ and a variable $t \leftarrow 0$. 
$\C$ runs $(\vk, \sk_{0}) \leftarrow \nPPSKGen (1^{\lambda})$ and sends $\vk$ to $\A$.

\item {\bf Query Phase:}
$\A$ makes the following puncturing queries and signing queries polynomially many times in an arbitrary order.
\begin{itemize}
\item  {\bf Puncturing Query:}
For a puncturing query on $\prf$, if $\prf \in  \bbL^{\Punc}$, $\C$ returns $\bot$.
Otherwise $\C$ updates $\bbL^{\Punc} \leftarrow \bbL^{\Punc} \cup \{\prf\}$, $t \leftarrow t + 1$, $\sk_{t} \leftarrow \nPPSPunc (\vk, \sk_{t-1}, \prf)$.

\item  {\bf Signing Query:}
For a signing query on $\msg = \prf||\sff$, if $\prf \in \bbL^{\Punc}$, $\C$ returns $\bot$. 
Otherwise $\C$ updates $\bbL^{\Sign} \leftarrow \bbL^{\Sign} \cup \{\msg\}$, runs $\sigma \leftarrow \nPPSSign (\vk, \sk_{t}, \msg)$, and returns $\sigma$ to $\A$.
\end{itemize}

\item {\bf Challenge Phase:} $\A$ outputs a target prefix $\prf^{*}$.
After outputting $\prf^{*}$, $\A$ makes the puncturing queries and signing queries as described in the query phase polynomially many times in an arbitrary order.

\item {\bf Corruption Phase:} If $\prf^{*} \in \bbL^{\Punc}$, $\C$ sends a current signing key $\sk_{t}$ to $\A$.
Otherwise, $\C$ sends $\bot$ to $\A$.

\item {\bf Finalization:}  
$\A$ finally outputs a forgery $(\msg^{*}= \prf^{*}||\sff^{*}, \sigma^{*})$ to $\C$. 
If $\prf^{*} \in \bbL^{\Punc} \land \msg^{*} \notin \bbL^{\Sign} \land \nPPSVerify (\vk, \msg^{*}, \sigma^{*}) = 1$ holds, return $1$. 
Otherwise, return $0$.
\end{itemize}
The advantage of $\A$ is defined by $\Adv^{\EUFAPCMA}_{\PPS, \A}(\lambda):= \Pr[\sfGame^{\EUFAPCMA}_{\PPS, \A}(1^{\lambda}) \Rightarrow 1]$.
We say that $\PPS$ satisfies the EUF-AP-CMA security if for any PPT adversary $\A$, $\Adv^{\EUFAPCMA}_{\PPS, \A}(\lambda)$ is $\negl(\lambda)$.
\end{definition}

\subsection{Security Revisited and Our Security Model}
Here, we recall the security definition of a prefix puncturable signature scheme given in Definition \ref{Def_PPS_EUF-AP-CMA}. 
In this definition, an adversary $\A$ in the EUF-AP-CMA game must submit a target prefix $\prf^{*}$ in the challenge phase. 
To win the EUF-AP-CMA game, $\A$ is required to output a forgery for $\prf^{*}$. 
However, requiring the adversary to specify $\prf^{*}$ before obtaining the punctured signing key seems unnatural, since this requirement prevents the adversary from choosing its target based on the information contained in the exposed signing key. 

This distinction can also be viewed as analogous to selective and adaptive identity selection in identity-based cryptography: regarding a prefix as playing a role analogous to an identity, it is natural to consider both non-adaptive and adaptive choices of the target prefix. 
Instead, it appears more natural to remove the challenge phase from the EUF-AP-CMA game and allow $\A$ to output a forgery for any prefix $\prf^{*}$ on which a puncturing query has been made in the finalization phase.

Motivated by this observation, we propose a new security notion called existential unforgeability under chosen message attacks with adaptive puncturing and adaptive target prefix (EUF-AP-ATP-CMA) security. 
Our security model allows the adversary to choose the target prefix when producing the final forgery, after observing the punctured signing key, and therefore models a more adaptive adversary.

\begin{definition}[EUF-AP-ATP-CMA Security (Our Proposal)]
Let $\PPS = (\nPPSKGen, \allowbreak \nPPSPunc, \allowbreak \nPPSSign, \nPPSVerify)$ be a prefix puncturable signature scheme and $\A$ be a PPT adversary.
The existential unforgeability under chosen message attacks with adaptive puncturing and adaptive target prefix (EUF-AP-ATP-CMA) security is defined via the following EUF-AP-ATP-CMA game $\sfGame^{\EUFAPATPCMA}_{\PPS, \A}(1^{\lambda})$ between a challenger $\C$ and the adversary $\A$. 

\begin{itemize}
\item {\bf Initial Setup:}
$\C$ initializes lists $\bbL^{\Sign} \leftarrow \{\}$, $\bbL^{\Punc} \leftarrow \{\}$, and a variable $t \leftarrow 0$.
$\C$ runs $(\vk, \sk_{0}) \leftarrow \nPPSKGen (1^{\lambda})$ and sends $\vk$ to $\A$.

\item {\bf Query Phase:}
$\A$ makes the following puncturing queries and signing queries polynomially many times in an arbitrary order.
\begin{itemize}
\item  {\bf Puncturing Query:}
For a puncturing query on $\prf$, if $\prf \in \bbL^{\Punc}$, $\C$ returns $\bot$.
Otherwise $\C$ updates $\bbL^{\Punc} \leftarrow \bbL^{\Punc} \cup \{\prf\}$, $t \leftarrow t + 1$, $\sk_{t} \leftarrow \nPPSPunc (\vk, \sk_{t-1}, \prf)$.

\item  {\bf Signing Query:}
For a signing query on $\msg = \prf||\sff$, if $\prf \in \bbL^{\Punc}$, $\C$ returns $\bot$. 
Otherwise $\C$ updates $\bbL^{\Sign} \leftarrow \bbL^{\Sign} \cup \{\msg\}$, runs $\sigma \leftarrow \nPPSSign (\vk, \sk_{t}, \msg)$, and returns $\sigma$ to $\A$.
\end{itemize}

\item {\bf  Corruption Phase:} $\A$ outputs an instruction $\inscorrupt$ to $\C$.
Then, $\C$ sends a current signing key $\sk_{t}$ to $\A$.

\item {\bf Finalization:}  
$\A$ finally outputs a forgery $(\msg^{*}= \prf^{*}||\sff^{*}, \sigma^{*})$ to $\C$. 
If $\prf^{*} \in \bbL^{\Punc} \land \msg^{*} \notin \bbL^{\Sign} \land \nPPSVerify (\vk, \msg^{*}, \sigma^{*}) = 1$ holds, return $1$.
Otherwise, return $0$.
\end{itemize}

The advantage of $\A$ is defined by $\Adv^{\EUFAPATPCMA}_{\PPS, \A}(\lambda) \allowbreak:=$ $\Pr[\sfGame^{\EUFAPATPCMA}_{\PPS, \A}(1^{\lambda}) \Rightarrow 1]$.
We say that $\PPS$ satisfies the EUF-AP-ATP-CMA security if for any PPT adversary $\A$, $\Adv^{\EUFAPATPCMA}_{\PPS, \A}(\lambda)$ is $\negl(\lambda)$.
\end{definition}

Here, we clarify the fact that the EUF-AP-ATP-CMA security implies the EUF-AP-CMA security.

\begin{theorem}\label{Th_AP-ATP_imply_AP_security}
Let $\PPS = (\nPPSKGen, \allowbreak \nPPSPunc, \allowbreak \nPPSSign, \nPPSVerify)$ be a prefix puncturable signature scheme and $\A$ be a PPT adversary against the EUF-AP-CMA security of $\PPS$.
Then, there is a PPT adversary $\R$ against the EUF-AP-ATP-CMA security of $\PPS$ that satisfies
\begin{equation*}
\Adv^{\EUFAPCMA}_{\PPS, \A}(\lambda) = \Adv^{\EUFAPATPCMA}_{\PPS, \R}(\lambda).
\end{equation*}
\end{theorem}

\begin{proof}
We prove Theorem \ref{Th_AP-ATP_imply_AP_security} by assuming the existence of a PPT adversary $\A$ that breaks the EUF-AP-CMA security of $\PPS$, and constructing a PPT adversary $\R$ that breaks the EUF-AP-ATP-CMA security of $\PPS$.
Let $\C$ be the challenger of the EUF-AP-ATP-CMA security of $\PPS$.
The construction of $\R$ is given as follows.

\begin{itemize}
\item $\R$ takes an instance $\vk$ of the EUF-AP-ATP-CMA security game.
Then, $\R$ initializes lists $\bbL^{\Sign} \leftarrow \{\}$, $\bbL^{\Punc} \leftarrow \{\}$ and invokes $\A$ on input $\vk$.

\item For a puncturing query $\prf$ from $\A$,  if $\prf \in \bbL^{\Punc}$, $\R$ returns $\bot$.
Otherwise $\R$ updates $\bbL^{\Punc} \leftarrow \bbL^{\Punc} \cup \{\prf\}$ and makes a puncturing query to~$\C$.

\item For a signing query on $\msg = \prf||\sff$, if $\prf \in \bbL^{\Punc}$, $\R$ returns $\bot$. 
Otherwise $\R$ updates $\bbL^{\Sign} \leftarrow \bbL^{\Sign} \cup \{\msg\}$, makes a signing query with $\msg = \prf||\sff$, receives a signature $\sigma$. 
Then, $\R$ returns $\sigma$ to $\A$.

\item In the challenge phase of $\A$, $\R$ receives a target prefix $\prf^{*}$ and stores it.
For puncturing queries and signing queries after receiving $\prf^{*}$, $\R$ responds to these queries the same way as described above.

\item In the corruption phase of $\A$, if $\prf^{*} \in \bbL^{\Punc}$, $\R$ sends an instruction $\inscorrupt$ to $\C$, receives $\sk_{t}$, and returns $\sk_{t}$ to $\A$.
Otherwise, $\R$ sends $\bot$ to $\A$.
%Otherwise, $\R$ sends an instruction $\ins =  \insskip$ to $\C$.

\item After receiving the final output $(\msg^{*}= \prf^{*}||\sff^{*}, \sigma^{*})$ from $\A$, $\R$ outputs $(\msg^{*}= \prf^{*}||\sff^{*}, \sigma^{*})$ as the final output.
\end{itemize}

It is clear that $\R$ perfectly simulates the challenger of the EUF-AP-CMA game.
If $\A$ outputs a valid forgery $(\msg^{*}= \prf^{*}||\sff^{*}, \sigma^{*})$ for the EUF-AP-CMA game, then $\R$ also outputs a valid forgery for the EUF-AP-ATP-CMA game.
Thus, we conclude Theorem \ref{Th_AP-ATP_imply_AP_security}. \qed
\end{proof}

We prove the reverse implication (i.e., the EUF-AP-CMA security implies the EUF-AP-ATP-CMA security).
Note that we prove this claim via a non-tight reduction.

\begin{theorem}\label{Th_AP_imply_AP-ATP_non-tight_security}
Let $\PPS = (\nPPSKGen, \allowbreak \nPPSPunc, \allowbreak \nPPSSign, \nPPSVerify)$ be a prefix puncturable signature scheme and $\A$ be a PPT adversary against the EUF-AP-ATP-CMA security of $\PPS$ that makes $Q^{\Punc}$ puncturing queries.
Then, there is a PPT adversary $\R$ against the EUF-AP-CMA security of $\PPS$ that satisfies
\begin{equation*}
\Adv^{\EUFAPATPCMA}_{\PPS, \A}(\lambda) \leq Q^{\Punc}  \cdot  \Adv^{\EUFAPCMA}_{\PPS, \R}(\lambda).
\end{equation*}
\end{theorem}

\begin{proof}
We prove Theorem \ref{Th_AP_imply_AP-ATP_non-tight_security} by assuming the existence of a PPT adversary $\A$ that breaks the EUF-AP-ATP-CMA security of $\PPS$, and constructing a PPT adversary $\R$ that breaks the EUF-AP-CMA security of $\PPS$.
Let $\C$ be the challenger of the EUF-AP-CMA security of $\PPS$.
The construction of $\R$ is given as follows.

\begin{itemize}
\item $\R$ takes an instance $\vk$ of the EUF-AP-CMA security game.
$\R$ initializes lists $\bbL^{\Sign} \leftarrow \{\}$, $\bbL^{\Punc} \leftarrow \{\}$. 
Then, $\R$ invokes $\A$ on input $\vk$.

\item For a puncturing query $\prf$ from $\A$,  if $\prf \in \bbL^{\Punc}$, $\R$ returns $\bot$.
Otherwise $\R$ updates $\bbL^{\Punc} \leftarrow \bbL^{\Punc} \cup \{\prf\}$ and makes a puncturing query to~$\C$.

\item For a signing query on $\msg = \prf||\sff$, if $\prf \in \bbL^{\Punc}$, $\R$ returns $\bot$. 
Otherwise $\R$ updates $\bbL^{\Sign} \leftarrow \bbL^{\Sign} \cup \{\msg\}$, makes a signing query with $\msg = \prf||\sff$, receives a signature $\sigma$. 
Then, $\R$ returns $\sigma$ to $\A$.

\item  In the corruption phase of $\A$, if $\R$ receives $\inscorrupt$, $\R$ chooses a guessed target prefix $\widetilde{\prf}^{*} \xleftarrow{\$} \bbL^{\Punc}$, moves to the challenge phase of $\R$, sends $\widetilde{\prf}^{*}$ to $\C$, moves to the corruption phase of $\R$, receives $\sk_{t}$, and returns $\sk_{t}$ to $\A$.

\item $\R$ receives the final output $(\msg^{*}= \prf^{*}||\sff^{*}, \sigma^{*})$ from $\A$. 
If $\widetilde{\prf}^{*}  = \prf^{*}$ holds, $\R$ returns $(\msg^{*}= \prf^{*}||\sff^{*}, \sigma^{*})$ to $\C$.
Otherwise, $\R$ aborts.

\end{itemize}

It is clear that $\R$ perfectly simulates the challenger of the EUF-AP-ATP-CMA game.
If $\A$ outputs a valid forgery $(\msg^{*}= \prf^{*}||\sff^{*}, \sigma^{*})$ for the EUF-AP-ATP-CMA game, then $\R$ outputs a valid forgery if $\widetilde{\prf}^{*} = \prf^{*}$ holds.
The probability that $\widetilde{\prf}^{*} = \prf^{*}$ holds is at least $\frac{1}{Q^{\Punc}}$.

From this fact, we see that
\begin{equation*}
\Pr[\sfGame^{\EUFAPCMA}_{\PPS, \R}(1^{\lambda}) \Rightarrow 1] \geq  \frac{1}{Q^{\Punc}} \cdot \Pr[\sfGame^{\EUFAPATPCMA}_{\PPS, \A}(1^{\lambda}) \Rightarrow 1]
\end{equation*}
holds.
Thus, we conclude Theorem~\ref{Th_AP_imply_AP-ATP_non-tight_security}. \qed
\end{proof}

By Theorem~\ref{Th_AP-ATP_imply_AP_security} and Theorem~\ref{Th_AP_imply_AP-ATP_non-tight_security}, the notions of EUF-AP-ATP-CMA security and EUF-AP-CMA security are equivalent in the sense that reductions exist in both directions. 
However, this equivalence does not appear to be tight, since the reduction incurs a loss factor of $Q^{\Punc}$, where $Q^{\Punc}$ denotes the number of puncturing operations applied to the secret key.

\section{Prefix Puncturable Signatures from HIBS}\label{Sect_PPS_from_HIBS}
In this section, we review the complete subtree algorithm $\CS$ which is used for our construction.
Then, we present a generic construction of a prefix puncturable signature scheme $\PPS_{\Ours}$ from a hierarchical identity-based signature scheme.

\subsection{Complete Subtree Algorithm}\label{SubSect_Comp_Subtree_Algo}
We use the complete subtree algorithm $\CS$ \cite{NNL01} for our construction.
We consider a complete binary tree $T$ of level $\ell$ whose nodes are labeled with binary strings.
The root (i.e., the node of level $0$) is labeled with the empty string $\epsilon$.
For every node at level $i$ labeled by $x \in \{0, 1\}^{i}$, its left and right children are labeled by $x0$ and $x1$, respectively.

The deterministic algorithm $\CS$ takes as input a set $D$ of leaf labels and outputs a set $C$ of tree node labels such that:
\begin{itemize}
\item Every leaf whose label is not in $D$ lies in the subtree of at least one node in~$C$.
\item No leaf whose label is in $D$ lies in the subtree of any node in $C$.
\end{itemize}
We illustrate an example of the input and output of this algorithm in Fig. \ref{Fig_Example_CS}.

\begin{figure}[htbp]
\centering
\begin{tikzpicture}[
    level distance=14mm,
    level 1/.style={sibling distance=48mm},
    level 2/.style={sibling distance=24mm},
    level 3/.style={sibling distance=12mm},
    internal/.style={
        circle,
        draw,
        minimum size=7mm,
        inner sep=0pt
    },
    leaf/.style={
        circle,
        draw,
        minimum size=7mm,
        inner sep=1pt
    },
    cover/.style={
        fill=red!40
    },
    excluded/.style={
        fill=blue!25
    }
]

\node[internal] {$\epsilon$}
child {
    node[internal,cover] {$0$}
    child {
        node[internal] {$00$}
        child { node[leaf] {$000$} }
        child { node[leaf] {$001$} }
    }
    child {
        node[internal] {$01$}
        child { node[leaf] {$010$} }
        child { node[leaf] {$011$} }
    }
}
child {
    node[internal] {1}
    child {
        node[internal] {10}
        child { node[leaf,excluded] {$100$} }
        child { node[leaf,cover] {$101$} }
    }
    child {
        node[internal] {11}
        child { node[leaf,cover] {$110$} }
        child { node[leaf,excluded] {$111$} }
    }
};

\end{tikzpicture}
\caption{\small An example of the input and output of $\CS$ for a complete binary tree. 
For example, suppose that $\ell = 3$ and $\CS$ takes \colorbox{blue!25}{$D = \{100, 111\}$} as input, corresponding to the leaves highlighted in blue in the figure. 
Then, $\CS$ outputs \colorbox{red!40}{$C = \{0, 101, 110\}$}, corresponding to the nodes highlighted in red.}
\label{Fig_Example_CS}
\end{figure}
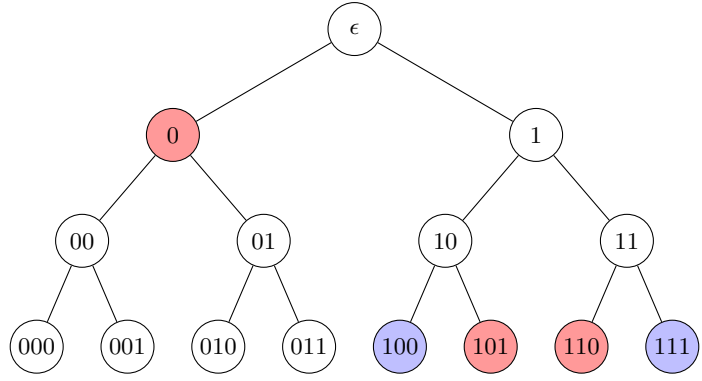

For a complete binary tree $T$ of level $\ell$ and any set $D$ of leaf labels, there exists a collection of $O(\ell |D|)$ subtrees that covers all leaves outside $D$ while excluding all leaves in $D$.
We also use the following property of the complete subtree algorithm $\CS$.
For $D\subseteq D'$, for every node $w' \in \CS(D')$, 
there exists a node $w \in \CS(D)$ such that either $w=w'$ or $w'$ is a descendant of $w$.

\subsection{Our Construction}
Let $\HIBS = (\HIBSSetup, \HIBSExtract, \allowbreak \HIBSSign, \HIBSVerify)$ be an $\ell$-level hierarchical identity-based signature scheme with an identity space $ID^{\HIBS} = ID_{1} \times ID_{2} \times \dots \times ID_{\ell}$ and message space $M^{\HIBS} = \{0, 1\}^{m}$.
To simplify the discussion, we assume that the identity space at each level $ID_{i}$ is $\{0,1\}$ (i.e., $ID^{\HIBS} = \{0, 1\}^{\ell}$).

Let $\CS$ be the complete subtree algorithm for a level $\ell$ binary tree.
%Moreover, we introduce the the non-covered-leaf algorithm $\NCL$ for our construction.
Let $T$ be a complete binary tree of level $\ell$ whose nodes are labeled with binary strings in the same way as described in Section~\ref{SubSect_Comp_Subtree_Algo}.
We refer to the following algorithm as the non-covered-leaf algorithm $\NCL$.
The algorithm $\NCL$ takes as input a set $C$ of node labels and outputs a set $D$ of leaf labels such that no leaf in $D$ belongs to any subtree rooted at a node in $C$.
We illustrate an example of the input and output of this algorithm in Fig. \ref{Fig_Example_NCL}.

\begin{figure}[htbp]
\centering
\begin{tikzpicture}[
    level distance=14mm,
    level 1/.style={sibling distance=48mm},
    level 2/.style={sibling distance=24mm},
    level 3/.style={sibling distance=12mm},
    internal/.style={
        circle,
        draw,
        minimum size=7mm,
        inner sep=0pt
    },
    leaf/.style={
        circle,
        draw,
        minimum size=7mm,
        inner sep=1pt
    },
    cover/.style={
        fill=red!40
    },
    excluded/.style={
        fill=blue!25
    }
]

\node[internal] {$\epsilon$}
child {
    node[internal] {$0$}
    child {
        node[internal,cover] {$00$}
        child { node[leaf] {$000$} }
        child { node[leaf] {$001$} }
    }
    child {
        node[internal] {$01$}
        child { node[leaf,excluded] {$010$} }
        child { node[leaf,excluded] {$011$} }
    }
}
child {
    node[internal] {1}
    child {
        node[internal] {10}
        child { node[leaf,cover] {$100$} }
        child { node[leaf,excluded] {$101$} }
    }
    child {
        node[internal,cover] {11}
        child { node[leaf] {$110$} }
        child { node[leaf] {$111$} }
    }
};

\end{tikzpicture}
\caption{\small An example of the input and output of $\NCL$. For example, suppose that $\ell = 3$ and $\NCL$ takes \colorbox{red!40}{$C = \{00, 11, 100\}$} as input, where $C$ corresponds to the nodes highlighted in red in the figure.
Then, $\NCL$ outputs \colorbox{blue!25}{$D = \{010, 011,101 \}$}, where $D$ corresponds to the leaves highlighted in blue.}
\label{Fig_Example_NCL}
\end{figure}
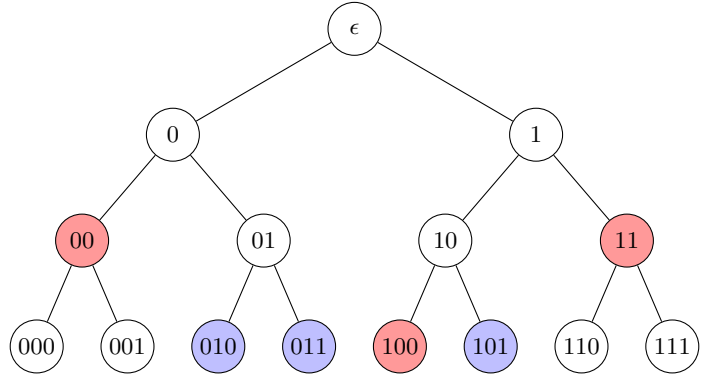

Now, we are ready to present our construction of a prefix puncturable signature scheme.
We give our construction $\PPS_{\Ours}[\HIBS]$ with a message space $M^{\PPS} =  P^{\PPS} \times S^{\PPS} = ID^{\HIBS} \times M^{\HIBS}$ in Fig. \ref{Fig_Our_Const}.

\begin{figure}[htbp]
\centering
\setlength{\extrarowheight}{2pt}
\begin{tabular}{|l|}
\hline
$\PPSKGen(1^{\lambda}):$\\
~~~$(\ppHIBS, \skHIBS_{\epsilon}) \leftarrow \HIBSSetup (1^{\lambda})$, $C \leftarrow \{\epsilon\}$. \\
~~~Return $(\vkPPS, \skPPS_{0}) \leftarrow (\ppHIBS, ((\skHIBS_{\epsilon}), C))$.\\
\hline
$\PPSPunc(\vkPPS, \skPPS_{i} = ((\skHIBS_{\w})_{\w \in C}, C), \prf_{i} \in \{0, 1\}^{\ell}):$\\
~~~$D \leftarrow \NCL(C)$, $D' \leftarrow D \cup \{\prf_{i}\}$, $C' \leftarrow \CS(D')$.\\
~~~Compute $(\skHIBS_{\w})_{\w \in C'}$ from $(\skHIBS_{\w})_{\w \in C}$ by running $ \HIBSExtract$.\\
~~~Return $\skPPS_{i+1} \leftarrow ((\skHIBS_{\w})_{\w \in C'}, C')$.\\
\hline
$\PPSSign(\vkPPS, \skPPS_{i}=((\skHIBS_{\w})_{\w \in C}, C), \msg=\prf||\sff):$\\
~~~$D \leftarrow \NCL(C)$.\\
~~~If $\prf \in \NCL(C)$, return $\bot$.\\
~~~Find $\w' \in \CS(D)$ such that $\prf \in I[\w']$,\\
~~~~~ where $I[\w']$ is a set of strings that have $\w'$ as a prefix. \\
~~~Derive $\skHIBS_{\prf}$ from $\skHIBS_{\w'}$ by running $\HIBSExtract$.\\
~~~Return $\sigma \leftarrow \HIBSSign (\ppHIBS, \skHIBS_{\prf}, \prf, \sff)$.\\
\hline
$\PPSVerify(\vkPPS = \ppHIBS, \msg=\prf||\sff, \sigma):$\\
~~~If $\HIBSVerify(\ppHIBS, \prf, \sff, \sigma) = 1$, return $1$.\\
~~~Otherwise, return $0$.\\
\hline
\end{tabular}
\caption{\small Our construction $\PPS_{Ours}[\HIBS]$.}
\label{Fig_Our_Const}
\end{figure}

\subsection{Analysis}

We analyze our construction $\PPS_{\Ours}[\HIBS]$.
The correctness of $\PPS_{\Ours}[\HIBS]$ follows directly from that of $\HIBS$.
The security of our construction is proven under the adaptive-ID security of $\HIBS$.

\begin{theorem}\label{Th_HIBS_imply_PPS}
If the scheme $\HIBS$ satisfies the EUF-AID-CMA security, then our construction $\PPS_{\Ours}[\HIBS]$ satisfies the EUF-AP-ATP-CMA security.

More precisely,  let $\A$ be a PPT adversary against the EUF-AP-ATP-CMA security of $\PPS_{\Ours}[\HIBS]$.
Then, there is a PPT reduction algorithm $\R$ against the EUF-AID-CMA security of $\HIBS$ that satisfies
\begin{equation*}
\Adv^{\EUFAPATPCMA}_{\PPS_{\Ours}, \A}(\lambda) = \Adv^{\EUFAIDCMA}_{\HIBS, \R}(\lambda).
\end{equation*}

\end{theorem}

\begin{proof}
Let $\A$ be a PPT adversary against the EUF-AP-ATP-CMA security of the scheme $\PPS_{\Ours}[\HIBS]$ and $\C$ be the challenger of the EUF-AID-CMA security game for $\HIBS$.
We prove Theorem \ref{Th_HIBS_imply_PPS} by constructing a reduction algorithm $\R$ against the EUF-AID-CMA security of the scheme $\HIBS$.
We describe the construction of $\R$ as follows.

\begin{itemize}
\item $\R$ takes an instance $\ppHIBS$ of the EUF-AID-CMA security game.
$\R$ sets $\bbL^{\Sign} \leftarrow \{\}$, $D \leftarrow \{\}$, $t \leftarrow 0$, $\vkPPS \leftarrow \ppHIBS$.
Then, $\R$ invokes $\A$ with the input $\vkPPS$.

\item For a puncturing query on $\prf$, if $\prf \in D$, $\R$ returns $\bot$.
Otherwise $\R$ updates $D \leftarrow D \cup \{\prf\}$.

\item For a signing query on $\msg = \prf||\sff$, if $\prf \in D$, $\R$ returns $\bot$. 
Otherwise $\R$ updates $\bbL^{\Sign} \leftarrow \bbL^{\Sign} \cup \{\msg\}$, makes a signing query on $(\prf, \sff)$ to $\C$, obtains a signature $\sigma$, and returns $\sigma$ to $\A$.

\item For an instruction $\inscorrupt$ from $\A$, $\R$ computes $C \leftarrow \CS(D)$.
Then, for each $\w \in C$, $\R$ makes a corruption query on $\w$ to $\C$ and obtains $\skHIBS_{\w}$.
$\R$ sets $\skPPS_{t} \leftarrow ((\skHIBS_{\w})_{\w \in \CS(D)}, C)$ and sends $\skPPS_{t}$ to $\A$.

\item After receiving the final output $(\msg^{*}= \prf^{*}||\sff^{*}, \sigma^{*})$ from $\A$, $\R$ outputs $(\widetilde{\id}^{*}=\prf^{*}, \widetilde{\msg}^{*}=\sff^{*}, \sigma^{*})$ as the final output.
\end{itemize}

Clearly, $\R$ perfectly simulates the challenger of the EUF-AP-ATP-CMA game.
It remains to show that a valid forgery $(\msg^{*}= \prf^{*}||\sff^{*}, \sigma^{*})$  for the EUF-AP-ATP-CMA game of $\PPS_{\Ours}[\HIBS]$ produced by $\A$ yields a valid forgery for the EUF-AID-CMA game of $\HIBS$.

Suppose that $\A$ outputs a valid forgery $(\msg^{*} = \prf^{*}||\sff^{*}, \sigma^{*})$ for the EUF-AP-ATP-CMA game of $\PPS_{\Ours}[\HIBS]$.
Since $\A$ wins the game, we have $\prf^{*} \in D$ and $\msg^{*} = \prf^{*} || \sff^{*} \notin L^{\Sign}$.
In the corruption phase, $\R$ makes corruption queries on every $\w \in C = \CS(D)$.
By the definition of $\CS(D)$, no leaf in $D$ belongs to a subtree rooted at a node in $C$.
Since $\prf^{*} \in D$, we have $\prf^{*} \notin I[\w]$ for every $\w \in C$.
Hence, the corruption queries made by $\R$ do not violate the winning condition of the EUF-AID-CMA game for the target identity $\prf^{*}$.
Moreover, since $\prf^{*} || \sff^{*} \notin L^{\Sign}$, $\R$ has never made a signing query on $(\prf^{*}, \sff^{*})$ to~$\C$.
This fact implies that $(\prf^{*}, \sff^{*})$ is not contained in the signing-query list of the EUF-AID-CMA game.
Finally, if $\PPSVerify(\vkPPS, \prf^{*}||\sff^{*},\sigma^{*})=1$ holds, then $\HIBSVerify(\ppHIBS, \allowbreak \prf^{*}, \sff^{*}, \sigma^{*}) = 1$ holds.
Therefore, $(\prf^{*},\sff^{*},\sigma^{*})$ is a valid forgery for the EUF-AID-CMA game of $\HIBS$.

Then, we see that
\begin{equation*}
\Pr[\sfGame^{\EUFAIDCMA}_{\HIBS, \R}(1^{\lambda}) \Rightarrow 1]  = \Pr[\sfGame^{\EUFAPATPCMA}_{\PPS_{\Ours}, \A}(1^{\lambda}) \Rightarrow 1]
\end{equation*}
holds.
Thus, we conclude Theorem \ref{Th_HIBS_imply_PPS}.
\qed
\end{proof}

\section*{Acknowledgement}
A part of this work was supported by JSPS KAKENHI JP24H00071, JST CREST JPMJCR2113, and JST K Program JPMJKP24U2.

\bibliographystyle{abbrvurl}
\bibliography{PuncS}

\newpage
\setcounter{tocdepth}{2}
\tableofcontents

\end{document}